\documentclass[twocolumn]{article}
\usepackage[utf8]{inputenc}
\usepackage{listings}
\usepackage{color}
\usepackage{fullpage}
\usepackage{color}
\usepackage{amssymb,graphicx}
\usepackage{float}
\usepackage{verbatim}
\floatstyle{plain}
\restylefloat{figure}
\usepackage{caption}
\usepackage{subcaption}
\usepackage{tabu}
\usepackage{graphicx}
\usepackage{hyperref}
\usepackage{amsthm}
\usepackage{amsmath}
\usepackage{algorithm}
\usepackage[noend]{algpseudocode}
\usepackage{tikz}
\usepackage{makecell}
\usepackage{comment}
\usepackage{mathtools}
\usepackage{url}

\theoremstyle{definition}

\theoremstyle{definition}

\newtheorem{theorem}{Theorem}[subsection]
\newtheorem{lemma}{Lemma}[subsection]
\theoremstyle{definition}

\newcommand{\bDiamond}{\mathbin{\Diamond}}

\date{}
\title{A Generic Efficient Biased Optimizer for Consensus Protocols}

	\author{Yehonatan Buchnik and Roy Friedman\\
		Computer Science Department\\
		Technion\\
		\texttt{\{yon\_b,roy\}@cs.technion.ac.il}}

\begin{document}
\maketitle

\begin{abstract}
Consensus is one of the most fundamental distributed computing problems.
In particular, it serves as a building block in many replication based fault-tolerant systems and in particular in multiple recent blockchain solutions.
Depending on its exact variant and other environmental assumptions, solving consensus requires multiple communication rounds.
Yet, there are known optimistic protocols that guarantee termination in a single communication round under favorable conditions.

In this paper we present a generic optimizer than can turn any consensus protocol into an optimized protocol that terminates in a single communication round whenever all nodes start with the same predetermined value and no Byzantine failures occur (although node crashes are allowed).
This is regardless of the network timing assumptions and additional oracle capabilities assumed by the base consensus protocol being optimized.

In the case of benign failures, our optimizer works whenever the number of faulty nodes $f<n/2$.
For Byzantine behavior, our optimizer's resiliency depends on the validity variant sought.
In the case of classical validity, it can accommodate $f<n/4$ Byzantine failures.
With the more recent \emph{external validity function} assumption, it works whenever $f<n/3$.
Either way, our optimizer only relies on oral messages, thereby imposing very light-weight crypto requirements.
\end{abstract}
	
\section{Introduction}

Consensus~\cite{Lamport1982} is one of the most studied problems in distributed computing~\cite{AW98,Raynal2018book}.
This is because any kind of coordination requires reaching an agreement.
Also, in replicated systems, nodes generally need to agree on the order in which updates are applied to the system in order to maintain a consistent state.
Yet, despite its simplicity and intuitive appeal, several impossibility results daunt researchers and designers of distributed systems.
First and foremost, the seminal FLP results proves that consensus is not solvable in asynchronous distributed environments in which a single process may fail by crashing~\cite{Fischer1985}.
Also, lower bounds have been shown on the number of communication rounds needed to solve consensus~\cite{KeidarR03,Martin2006} as well as inherent tradeoffs between the availability and consistency of distributed systems~\cite{FB96}.

Yet, it has been noted that in various settings some conditions are more likely to hold than others.
This has led to the development of various optimistic consensus protocols that terminate quickly when certain favorable conditions hold.
These may include lack of failures, or at least lack of certain types of failures, periods of network synchrony in an asynchronous network, as well as the composition of the proposed values.

In this work we explore the impact of assuming that one specific value is more likely to be proposed by all correct nodes.
For example, in several blockchain protocols, the protocol can be viewed as repeatedly having a leader proposing a block and then all nodes running a binary consensus protocol for deciding whether to accept this value~\cite{tendermint,BF2019full,Crain2017}.
Hence, in the ``common case'' where there are no failures and the network behaves in a synchronous manner during that block generation process, all nodes will vote to accept the proposed block.
Similarly, in a control system one can expect that under normal operating conditions all correct replicas will usually propose taking the same course of action as all are exposed to the same system state and sensor readings~\cite{B777b}.

To that end, we develop a generic optimizer for consensus protocols which \emph{always} terminates in a single communication round when all nodes propose the same a-priory preferred value $v$ and no Byzantine failures manifest (yet, they can be tolerated).
Specifically, our optimizer can be applied to both the benign failures and Byzantine failures models and is independent of any other network and oracle assumptions needed by the base consensus protocol it optimizes.
In the case of benign failures, we only require that the number of faulty nodes is bounded by $f<n/2$.
Here, whenever all nodes propose the same predefined value, the optimizer will always terminate in a single all-to-all communication~round.

For the Byzantine failure model, we distinguish between two definitions of the consensus problem, one that is based on classical Byzantine validity vs. the more recently proposed external validity definition~\cite{cachin-survey}.
For classical validity, our optimizer requires $f<n/4$ and always terminates in one communication round whenever all nodes propose the same anticipated value and no Byzantine failures occur.
Further, when Byzantine failure do occur, the termination and safety of the overall construction is upheld.

When the consensus definition requires external validity, our optimizer can tolerate $f<n/3$.
Under the assumption that either there is only two possible valid proposed values (binary consensus), or that the base protocol satisfies classical validity, our optimizer always terminates in one communication round whenever all nodes propose the same anticipated value and no Byzantine failures occur.
As before, the overall correctness is maintained even when Byzantine failures do occur.

Note that sometimes the external validity function requires proposed values to carry a cryptocraphic proof.
As cryptographic proofs are expected to be much larger than the value itself, we can further reduce the communication cost of the protocol in the optimistic phase by only broadcasting the values themselves in the first phase.
Only if no decision can be made after the first phase, then the proofs are being broadcast so they can be fed to the validity functions.
This means an additional communication phase when the first optimistic phase fails, but significantly smaller messages when the first phase succeeds.
Hence, this last optimization is beneficial in situations where the first phase is likely to succeed in the vast majority of consensus invocations (and the proofs are substantially larger than values).

One critical requirement when dealing with Byzantine failures is to disallow impersonation, which is often ensured through cryptography.
There are two common models for preventing impersonation, namely \emph{oral messages} and \emph{signed messages}.
Oral messages can only guarantee to a receiver of a message which node has sent the message.
Hence, it can be implemented using MAC tags.
In practice, this can be made transparent to the application by utilizing transport level secure communication such as TLS, HTTPS, etc.
Signed messages, on the other hand, include a publicly verifiable signature of the sender on the value being sent.
This enables one node to prove to another node that it has received a certain value from a given third party.
Hence, the signature must be explicitly generated by the application's code, and requires either private/public key signatures and verification, or attaching a vector of MAC tags, one for each node, on each message.
In other words, signed messages are more expensive and cumbersome to use.
Our optimizer only requires oral messages when dealing with Byzantine failures.
Hence, its cryptography related space and CPU time are minimal.

Finally, we also show a lower bound on the minimal resiliency required to terminate in one communication round in the Byzantine failure model subject to classical validity.
Specifically, we show that for classical validity there is no asynchronous Byzantine consensus protocol that can ensure termination in one communication round after receiving only $n-f$ messages when $n=3f+1$.

\paragraph*{Summary of contributions:} We develop a generic optimizer that ensures fast termination in one communication round when all processes start with the same value and Byzantine failures do not manifest.
We instantiate the optimizer for both the benign and Byzantine failure model with optimal resiliency and formally prove their correctness.
We also show a resiliency lower bound for such fast termination in the Byzantine failure model for solving consensus with classical validity.

\paragraph*{Paper Organization}
The rest of this paper is organized as follows: We survey other optimized consensus protocols in Section~\ref{sec:related}.
The formal model and problems statements appear in Section~\ref{sec:prelim}.
We present our generic optimizer in Section~\ref{sec:biased} and show the lower bound in Section~\ref{sec:lower}. 
Finally, we conclude with a discussion in Section~\ref{sec:discuss}.

\section{Related Work}
\label{sec:related}

The first work to explore one communication round consensus in the benign failure model is~\cite{BBGMR01}.
The basic protocol in~\cite{BBGMR01} requires $f<n/3$.
That protocol is also extended to support a preferred value which improves the resiliency requirement to $f<n/2$, similar to our work.
The main contribution of this paper compared to~\cite{BBGMR01} is in our exploration of this problem under Byzantine failures and in the fact that we present a single generic optimizer for both failure models.

The work of~\cite{Friedman:2005} explored simple Byzantine consensus protocols that can terminate in a single communication round whenever all nodes start with the same value and certain failures do not manifest.
Yet, the probabilistic protocol of~\cite{Friedman:2005} required $f<n/5$ while their deterministic protocol needed $f<n/6$.
In contrast, our optimizer when instantiated to Byzantine failures can withstand up to $f<n/4$ with the classical validity definition (and $f<n/3$ with external validity, which was not explored in~\cite{Friedman:2005}).
This is due to biasing the consensus into preferring a certain value.
The price that we pay compared to~\cite{Friedman:2005} is that if all nodes start with the non-preferable value and the respective failures do not manifest, the protocol in~\cite{Friedman:2005} would terminate in a single communication step while our optimizer would have to invoke the full protocol.

Traditional deterministic Byzantine consensus protocols, most notably PBFT~\cite{Castro1999} require at least 3 communication rounds to terminate.
Multiple works that reduce this number have been published, each presenting a unique optimization.
The Q/U work presented a client driven protocol~\cite{AGGRW05} which enables termination in two communication rounds when favorable conditions are met.
Yet, its resiliency requirement is $f<n/5$, compared to our $f<n/4$ for the classical validity and $f<n/3$ for external validity.
The HQ work improved the resiliency of Q/U to $f<n/3$, yet does not perform well under high network load~\cite{CMLRS}.
Also, our optimizer is generic whereas Q/U and HQ are specialized solutions, each tailored to its intricate protocol.

The Fast Byzantine Consensus (FaB) protocol was the first protocol to implement a Byzantine consensus protocol that terminates in two communication phases in the normal case while requiring $f<n/5$~\cite{Martin2006}.
The normal case in~\cite{Martin2006} is defined as when there is a unique correct leader, all correct acceptors agree on its identity, and the system is in a period of synchrony.
This protocol translates into a $3$ phase state machine replication protocol.
Another variant can accommodate $n\geq 3f + 2b + 1$, where $b\leq f$ is the upper bound on non-leaders suffering Byzantine failures.

Zyzzyva is a client driven protocol~\cite{Kotla2007} which terminates after $3$ communication rounds (including the communication between the client and the replicas) whenever the client receives identical replies from all $3f+1$ replicas.
Our optimizer obtains termination in a single communication round among the replicas even when upto $f$ of them may crush or be slow.
This is by relying on all-to-all communication, and ensuring fast termination only when the preferred value is included in the first $n-f$ replies.
Also, our optimizer is generic while Zyzzyva and FaB are specialized solutions.

A generic construction for optimized Byzantine consensus protocols appears in~\cite{Guerraoui2010} along with its Aliph and Azyzzyva instantiations.
The construction in~\cite{Guerraoui2010} enables switching between Byzantine consensus protocols depending on the changing environment conditions in a safe manner.
In particular, it is possible to switch from an optimistic fast protocol that fails to terminate to a recovery protocol that would ensure overall correctness.
Both Aliph and Azyzzyva are client driven protocols, which require receiving timely identical replies from all $3f+1$ replicas.
Hence, our optimizer can terminate quickly in worse environment conditions if the preferred value is also the one supported by most nodes.

The condition based approach for solving consensus identifies various sets of input values that enable solving consensus fast~\cite{Mostefaoui2003}.
This is by treating the set of input values held by all processes as an input vector to the problem.
Specifically, the work in~\cite{FriedmanMRR07} showed that when the possible input vectors correspond to error correcting codes, then consensus is solvable in a single communication round regardless of synchrony assumptions.



\section{Preliminaries}
\label{sec:prelim}

\paragraph*{Basics:}
We consider a distributed system consisting of $n$ nodes communicating by sending messages over a network.
We assume that the network maintains \emph{integrity}, meaning that messages delivered by the network were indeed sent by their claimed sender and that their content is not corrupted.
The level of \emph{reliability} of the network defines which portions and under which circumstances messages sent between two nodes must be delivered by the network.
Additionally, the level of \emph{synchrony} is the system defines if and what type of bounds exist on the latency between sending a message and its recipient and whether nodes have access to clocks and if so, how synchronized these clocks are.
Our optimizer mechanism only requires that during its optimizing phase, messages transmitted between two correct nodes (to be defined shortly) are eventually delivered.
In particular, we do not assume anything about the synchrony of the system.
When the base protocol we are optimizing is invoked, the network reliability and synchronization assumptions under which that protocol was designed to work must hold.

\paragraph*{Failure Modesl and Resiliency:}
Out of the $n$ nodes in the system, up to $f$ nodes may be \emph{faulty} while the others are \emph{correct}.
In the benign failure model, nodes may only fail by \emph{crashing}, i.e., they might halt their execution.
In the Byzantine failures model, faulty nodes may deviate arbitrarily from their protocol.
Yet, given our integrity assumptions, nodes may not impersonate each other.
In particular, we assume that there are no \emph{Sybil} attacks.
The ratio between $f$ and $n$ is known as the \emph{resiliency} level of the system.
The level of synchrony, reliability, failure model, and the problem being addressed all impact the maximal level of resiliency that can be obtained~\cite{Raynal2018book}.

Cryptography is a common way to thwart impersonation.
Here we can distinguish between the \emph{oral messages} model and the \emph{signed messages} model.
Oral messages can only ensure to the receiver that the message it received from the network was indeed sent by its claimed sender.
In practice, this can be implemented using MACs, which are relatively cheap, and is often made transparent to the protocol's code by relying on secure transport protocols such as TLS or HTTPS.
Signed messages enable verifying that a specific value was sent by a specific node, and hence can be used to securely pass values through intermediary nodes, or verify that the same value was sent to all nodes.
In practice, the signed messages model is typically implemented through asymmetric key cryptography~\cite{crypto-textbook} or attaching vectors of MAC tags~\cite{Castro1999}.
Our optimizer mechanism only requires oral messages (in the Byzantine failures model).

\paragraph*{Benign Consensus:}
As mentioned before, in this paper we address the consensus problem.
We start with the common definition of consensus in the benign failure model and then extend it to the Byzantine failure model.
In the consensus problem each node $p_i$ has an initial value $v_i$, also known as the proposed value of $p_i$.
The nodes must each decide on a value such that:
\begin{description}
	\item[Validity] Any decided value must be proposed by some process.
	\item[Agreement] All correct processes decide on the same value.
	\item[Termination] Every correct process eventually decides.
\end{description}

\paragraph*{Byzantine Consensus:}
When Byzantine failures are considered, the value proposed by a Byzantine process is somewhat meaningless, since a Byzantine node may propose an arbitrary value and even pretend to propose different values to different nodes.
Hence, a common approach for defining Byzantine consensus is to require:
\begin{description}
	\item[Byzantine Validity] If all correct processes propose the same value, then this is the only value that can be decided.
	\item[Byzantine Agreement] Same as agreement.
	\item[Byzantine Termination] Same as termination.
\end{description}

\paragraph*{External Validity Byzantine Consensus:}
A shortcoming of Byzantine validity is that whenever not all correct nodes start with the same value, any value may be decided.
Further, in some applications, such as blockchains, even a Byzantine node may propose a value that is valid w.r.t. the systems goals.
This leads to the definition of external validity Byzantine consensus~\cite{cachin-survey}.
Here, one assumes an external boolean validity function \texttt{valid}($v$) whose output is \textsc{true} iff $v$ is a valid proposed value and requires:
\begin{description}
	\item[Byzantine External Validity] Any decided value $v$ must satisfy \texttt{valid}($v$) = \textsc{true}.
	\item[Byzantine Agreement] Same as agreement.
	\item[Byzantine Termination] Same as termination.
\end{description}

\section{The Biased Optimizer}
\label{sec:biased}

We present the generic optimizer in Section~\ref{sec:generic}.
This is followed by the instantiates to the various failure models: the benign failure model is discussed in Section~\ref{sec:benign}, Byzantine failures with classic validity are handled in Section~\ref{sec:byzantine-classic}, while Byzantine failures with external validity are addressed in Section~\ref{sec:byzantine-external}.

\subsection{Generic Optimizer}
\label{sec:generic}

Algorithm~\ref{alg:gen} depicts the pseudocode for the generic optimizer algorithm from the point of view of node $p$.
Initially, each node broadcast its proposed value to all others (Line~\ref{OC:2})\footnote{
	Here the term broadcast is equivalent to sending the same message point-to-point to all other nodes.}.
Next, each node waits until it receives at least $n-f$ proposals from distinct nodes.
If all these values are the same as the biased value, the node decides $v$ and returns.
Otherwise, if the received proposals match an adoption criteria (\texttt{AdoptionCriteria}) for the value $v$, the local node invokes the standard consensus protocol with a value $v$ (Line~\ref{OC:7}).
Otherwise, the standard consensus protocol is invoked with the originally proposed value for the local node $u$ (Line~\ref{OC:9}).
Finally, if a node $p$ has already decided in the optimistic phase, but it notices that the standard consensus protocol has been invoked by another process (Line~\ref{OC:11}), then $p$ invokes the consensus with the value $v$ (the only one $p$ could have decided on).
This is in order to ensure that enough nodes participate in the standard consensus protocol to enable its termination.

\begin{algorithm}[t]
	\caption{$v$-Biased Optimizer -- code for $p$}
	\label{alg:gen}
	\footnotesize
	\begin{algorithmic}[1]
		\Procedure{OptimizedConsensus$_{v,\mathtt{Consensus},\mathtt{AdoptionCriteria}}$}{$v'$}
		\State $\mathit{broadcast}(v')$ \label{OC:2}
		\State \textbf{wait until} $n-f$ proposals $\hat{v}$ have been received \label{OC:3}
		\State $\mathit{votes}=$ \{received proposals\}
		\If {$\mathit{votes}=\{v\}$} 
		\State Decide $v$ \label{OC:5}
		\State Return $v$
		\EndIf
		\If {$\mathit{votes}$ matches \texttt{AdoptionCriteria}}
		\State Return \texttt{Consensus}.Propose($v$) \label{OC:7}
		\Else
		\State Return \texttt{Consensus}.Propose($v'$) \label{OC:9}
		\EndIf 
		\EndProcedure
		\\
		\If {invoked Decide $v$ and some node invoked \texttt{Consensus}.Propose($\hat{v}$)}
		\State \texttt{Consensus}.Propose($v$) \Comment{this is executed at most once} \label{OC:11}
		\EndIf
	\end{algorithmic}
\end{algorithm}

\subsection{Instantiation to Benign Failures}
\label{sec:benign}

\begin{figure}[t]
	\center{
		\includegraphics[scale=0.46]{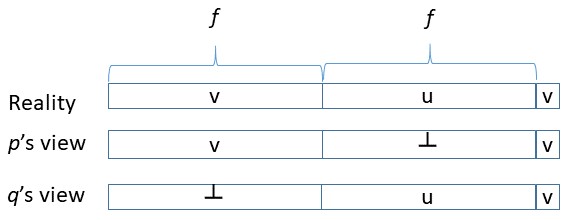}
	}
	\caption{A worst case scenario for the benign failure model and $f<n/2$. If some node $p$ decides after receiving $f+1$ votes for $v$, then any other node $q$ must receive at least one instance of $v$.}
	\label{fig:benign}
\end{figure}

For benign failures, the failure resiliency is $1/2-\epsilon$, that is, we assume $f<n/2$.
The \texttt{AdoptionCriteria} is simply to check whether the collection of received \textit{votes} includes at least one instance of the biased value $v$.
Here, if at least one node has received only (the preferred) $v$ values in Line~\ref{OC:3}, the only value it might decide on is $v$.
The worst that can happen is that all other $f$ nodes started with a different value, as illustrated in Figure~\ref{fig:benign}.
Yet, by the resiliency assumption, in this case any other node must have received at least one $v$ proposal.
Hence, all nodes would invoke the base consensus protocol with $v$ and by the validity and termination conditions of the consensus protocol would decide $v$.

In no node receives $n-f$ values of $v$, then it is clear that the standard consensus protocol would be invoked by each node with a value that was proposed by some node.
Hence, validity would be preserved as well.
The formal proof is given below.

\begin{lemma}
	\label{lemma:ben-termination}
	The protocol listed in Algorithm~\ref{alg:gen} satisfies termination as long as $f<n/2$.
\end{lemma}

\begin{proof}
	As can be seen from the code, as long as $f<n/2$, the protocol either decides and returns in Line~\ref{OC:5} or invokes the base consensus protocol in Lines~\ref{OC:7} or~\ref{OC:9}.
	If at least one node invokes the base consensus protocol, then by Line~\ref{OC:11} all nodes invoke this protocol (once).
	Hence, by the termination of the base consensus protocol, Algorithm~\ref{alg:gen} also terminates.
\end{proof}

\begin{lemma}
	\label{lemma:ben-validity}
	The protocol listed in Algorithm~\ref{alg:gen} satisfies validity.
\end{lemma}

\begin{proof}
	Clearly from the code, the value decided in Line~\ref{OC:5} is proposed by some (in fact by multiple) node(s).
	Similarly, if a value is adopted by the \texttt{AdoptionCriteria}, it is also proposed by some node.
	Hence, all invocations of the base consensus protocol in Lines~\ref{OC:7} or~\ref{OC:9} are with values proposed by some node.
	By the validity of the base consensus protocol, overall validity is satisfied.
\end{proof}

\begin{lemma}
	\label{lemma:ben-agreement}
	The protocol listed in Algorithm~\ref{alg:gen} satisfies agreement.
\end{lemma}

\begin{proof}
	Since there is only one preferred value $v$, if all nodes decide $v$ in Line~\ref{OC:5} then they all decide on the same value ($v$).
	Similarly, if all nodes decide by invoking the base consensus protocol in Lines~\ref{OC:7} or~\ref{OC:9}, then by the agreement property of this protocol all nodes decide on the same value.
	Hence, the only potential violation of agreement is if some nodes decide in Line~\ref{OC:5} while others invoke the base consensus protocol in Lines~\ref{OC:7} or~\ref{OC:9}.
	Suppose node $p$ decides in Line~\ref{OC:5}.
	In this case, it has received at least $n-f$ votes for $v$.
	Suppose all missing votes are $u$ and some other node $q$ has received these $u$ votes.
	However, as $f<n/2$, $q$ must receive at least $1$ vote for $v$ and therefore would adopt $v$ by the \texttt{AdoptionCriteria}\footnote{
		Notice that if in some execution the number of faulty nodes surpasses the assumed upper bound $f$, then the protocol would be stuck in Line~\ref{OC:3} and agreement would be trivially preserved.}.
	As this applies to any $q$, all nodes that invoke the base consensus protocol in Lines~\ref{OC:7} or~\ref{OC:9} invoke it with $v$.
	Hence, by the validity property of this protocol only $v$ can be decided so overall agreement is also satisfied.
\end{proof}

From Lemmas~\ref{lemma:ben-termination},~\ref{lemma:ben-validity} and~\ref{lemma:ben-agreement}, we immediately have the following theorem:
\begin{theorem}
	\label{thm:benign-correct}
	The protocol listed in Algorithm~\ref{alg:gen} solves consensus when up to $f<n/2$ nodes may fail by crashing.
\end{theorem}

\subsection{Byzantine Failures with Classical Validity}
\label{sec:byzantine-classic}

For the Byzantine failure model with classical validity, we provide a resiliency of $1/4-\epsilon$, i..e, $f<n/4$.
Here, the \texttt{AdoptionCriteria} is simply to check whether the collection of received \textit{votes} includes at least $f+1$ instances of the preferred value $v$.

\begin{figure}[t]
	\center{
		\includegraphics[scale=0.3]{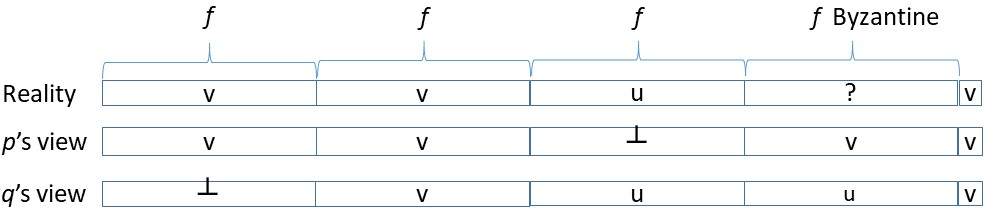}
	}
	\caption{A worst case scenario for the Byzantine failure model with classic validity and $f<n/4$. If some node $p$ decides after receiving $n-f$ votes for $v$, then any other node $q$ must receive at least $f+1$ votes for $v$.}
	\label{fig:byzclassic}
\end{figure}

Intuitively, if at least $f+1$ nodes broadcast $v$, then we know that at least one correct node proposed $v$ and therefore $v$ can be decided on in terms of validity.
On the other hand, if some node decided due to receiving $n-f$ values $v$, then it is possible that the other $f$ were simply late but all of them proposed value $u$.
In this case, another node might have received these $u$ votes first and might have received $u$ from all Byzantine nodes, collecting in total $2f$ votes for $u$.
This is illustrated in Figure~\ref{fig:byzclassic}.
Yet, as we assumed that $f<n/4$, such a node must receive at least $f+1$ votes for $v$ and would therefore adopt $v$ by the \texttt{AdoptionCriteria} and start the base consensus protocol with $v$.

In conclusion, all correct nodes would either decide $v$ in the optimizer phase, or would propose $v$ in the standard consensus algorithm.
In the latter case, due to its validity property, the decision value must also be $v$.
The complete formal proof appears below.

\begin{lemma}
	\label{lemma:byz-termination}
	The protocol listed in Algorithm~\ref{alg:gen} satisfies Byzantine termination as long as $f<n/4$.
\end{lemma}

\begin{proof}
	The proof is identical to the proof of Lemma~\ref{lemma:ben-termination}.
\end{proof}

\begin{lemma}
	\label{lemma:byz-validity}
	The protocol listed in Algorithm~\ref{alg:gen} satisfies Byzantine validity.
\end{lemma}

\begin{proof}
	Clearly from the code, the value decided by a correct node in Line~\ref{OC:5} must be proposed by at least one correct node given that $f<n/4$ and that in Line~\ref{OC:3} the node collects $n-f$ votes.
	In particular, if all correct nodes proposed the same value, this is the only value that can be decided at this stage.
	Similarly, if a value is adopted by the \texttt{AdoptionCriteria}, it is also proposed by some correct node.
	Hence, all invocations by correct nodes of the base consensus protocol in Lines~\ref{OC:7} or~\ref{OC:9} are with values proposed by some correct node.
	By the validity of the base consensus protocol, overall validity is satisfied.
\end{proof}

\begin{lemma}
	\label{lemma:byz-agreement}
	The protocol listed in Algorithm~\ref{alg:gen} satisfies Byzantine agreement when $f<n/4$.
\end{lemma}

\begin{proof}
	Since there is only one preferred value $v$, if all correct nodes decide $v$ in Line~\ref{OC:5} then they all decide on the same value ($v$).
	Similarly, if all correct nodes decide by invoking the base consensus protocol in Lines~\ref{OC:7} or~\ref{OC:9}, then by the agreement property of this protocol all correct nodes decide on the same value.
	Hence, the only potential violation of agreement is if some correct nodes decide in Line~\ref{OC:5} while others invoke the base consensus protocol in Lines~\ref{OC:7} or~\ref{OC:9}.
	Suppose correct node $p$ decides in Line~\ref{OC:5}.
	In this case, it has received at least $n-f$ votes for $v$.
	Suppose all missing votes are $u$ and some other node $q$ has received these $u$ votes.
	Further, the worst that can happen is that among the $n-f$ votes for $v$ that $p$ has received, $f$ were sent by Byzantine nodes who send the value $u$ to $q$.
	Putting it all together, as $f<n/4$, node $q$ must receive at least $f+1$ votes for $v$ and therefore would adopt $v$ by the \texttt{AdoptionCriteria}.
	As this applies to any $q$, all correct nodes that invoke the base consensus protocol in Lines~\ref{OC:7} or~\ref{OC:9} invoke it with $v$.
	Hence, by the validity property of this protocol only $v$ can be decided so overall agreement is also satisfied.
\end{proof}

From Lemmas~\ref{lemma:byz-termination},~\ref{lemma:byz-validity} and~\ref{lemma:byz-agreement}, we immediately have the following theorem:
\begin{theorem}
	\label{thm:byzantine-correct}
	The protocol listed in Algorithm~\ref{alg:gen} solves Byzantine consensus with classic validity when up to $f<n/4$ nodes may incur Byzantine failures.
\end{theorem}

\subsection{Byzantine Failures with External Validity}
\label{sec:byzantine-external}

We split the discussion of external validity into two parts.
First, we instantiate Algorithm~\ref{alg:gen} for this failure model in a manner that is oblivious to the proof that each value may carry, which might affect how the external validity function may decide if a value is valid.
This is done in Section~\ref{sec:oblivious}.
Next, in Section~\ref{sec:aware} we present another optimization that can reduce the amount of data being communicated in the optimistic case by separating between values and their proofs.

\subsubsection{Proof Oblivious Protocol}
\label{sec:oblivious}
Here, due to the use of an external validy function, we can raise the resiliency level to $f<n/3$.
This is because if the preferred value $v$ is valid, we can decide on it even if it was proposed by any node, correct or Byzantine.
We make the following three assumptions:
\begin{description}
	\item[Assumption 1] Any value $u$ proposed by a correct node is valid (both $u=v$ and $u\neq v$ are possible).
	\item[Assumption 2] The base consensus protocol satisfies classical validity (does not have to satisfy external validity).
	\item[Assumption 3] Either ($i$) there are only two possible proposed and decision values (known also as \emph{binary consensus}), or ($ii$) the base consensus protocol also satisfies external validity.
\end{description}
With these assumptions, the \texttt{AdoptionCriteria} becomes adopting $v$ if $v \in \mathit{votes}$ and $v$ is valid.

\begin{figure}[t]
	\center{
		\includegraphics[scale=0.38]{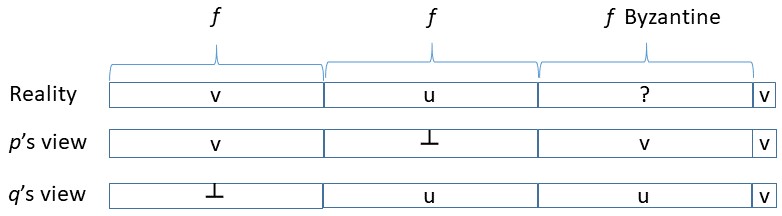}
	}
	\caption{A worst case scenario for the Byzantine failure model with external validity and $f<n/3$. If some node $p$ decides after receiving $n-f$ votes for $v$, then any other node $q$ must receive at least one vote for $v$.}
\label{fig:byzexternal}
\end{figure}

As shown in Figure~\ref{fig:byzexternal}, with $n=3f+1$ it is possible that some correct node receives $n-f$ instances of the preferred value $v$ and some other correct process only receives a single vote for $v$.
However, in such a case each correct process that does not decide in the optimizer phase must adopt $v$ and propose it to the standard consensus protocol.
Hence, the consensus protocol must decide $v$.
The complete formal proof follows.

\begin{lemma}
	\label{lemma:ext-termination}
	The protocol listed in Algorithm~\ref{alg:gen} satisfies Byzantine termination as long as $f<n/3$.
\end{lemma}

\begin{proof}
	The proof is identical to the proof of Lemma~\ref{lemma:ben-termination}.
\end{proof}

\begin{lemma}
	\label{lemma:ext-validity}
	The protocol listed in Algorithm~\ref{alg:gen} satisfies Byzantine external validity.
\end{lemma}

\begin{proof}
	Clearly from the code, the value decided by a correct node in Line~\ref{OC:5} must be proposed by at least one correct node given that $f<n/3$ and that in Line~\ref{OC:3} the node collects $n-f$ votes.
	By Assumption~1, $v$ is valid (w.r.t. external validity).
	Similarly, if the value $v$ is adopted by the \texttt{AdoptionCriteria}, it is also valid.
	Hence, all invocations by correct nodes of the base consensus protocol in Lines~\ref{OC:7} or~\ref{OC:9} are with valid values.
	
	We claim that from Assumption~3 (regarding the base consensus protocol), overall validity is also satisfied.
	Suppose part ($i$) of this assumption holds.
	Then if all correct nodes invoked the base consensus protocol with the same value (which we showed has to be valid), then only this value can be decided.
	Otherwise, some correct processes invoked consensus with one value and the others with another, both have to be valid, and these are the only two possible decision values. 
	Hence, only valid values can be decided on.
	
	In case part ($ii$) holds, then by the external validity of the base consensus protocol, overall validity is satisfied.
\end{proof}

\begin{lemma}
	\label{lemma:ext-agreement}
	The protocol listed in Algorithm~\ref{alg:gen} satisfies Byzantine agreement when $f<n/3$.
\end{lemma}

\begin{proof}
	As in the proof of Lemma~\ref{lemma:byz-agreement}, since there is only one preferred value $v$, if all correct nodes decide $v$ in Line~\ref{OC:5} then they all decide on the same value ($v$).
	Similarly, if all correct nodes decide by invoking the base consensus protocol in Lines~\ref{OC:7} or~\ref{OC:9}, then by the agreement property of this protocol all correct nodes decide on the same value.
	Hence, the only potential violation of agreement is if some correct nodes decide in Line~\ref{OC:5} while others invoke the base consensus protocol in Lines~\ref{OC:7} or~\ref{OC:9}.
	
	Suppose correct node $p$ decides in Line~\ref{OC:5}.
	In this case, it has received at least $n-f$ votes for $v$ meaning that at least $f+1$ correct nodes voted $v$.
	Suppose all missing votes are $u$ and some other node $q$ has received these $u$ votes.
	Further, the worst that can happen is that among the $n-f$ votes for $v$ that $p$ has received, $f$ were sent by Byzantine nodes who send the value $u$ to $q$.
	Putting it all together, as $f<n/3$, node $q$ must receive at least one vote for $v$ from a correct node, which by Assumption~1 means that $v$ is valid, and therefore $q$ would adopt $v$ by the \texttt{AdoptionCriteria}.
	As this applies to any $q$, all correct nodes that invoke the base consensus protocol in Lines~\ref{OC:7} or~\ref{OC:9} invoke it with $v$.
	Hence, by the validity property of this protocol (Assumption~2), only $v$ can be decided so overall agreement is also satisfied.
\end{proof}

From Lemmas~\ref{lemma:ext-termination},~\ref{lemma:ext-validity} and~\ref{lemma:ext-agreement}, we immediately have the following theorem:
\begin{theorem}
	\label{thm:external-correct}
	The protocol listed in Algorithm~\ref{alg:gen} solves Byzantine consensus with external validity when up to $f<n/3$ nodes may incur Byzantine failures.
\end{theorem}

\subsubsection{Proof Aware Variant}
\label{sec:aware}
We now explore another communication optimization when the validity function depends on the cryptographic proof that is part of the value.
That is, the initial value of each node is composed of $v = \{\mathit{val},\mathit{proof}\}$.
As the proof is large and would rarely be used, in the first phase we can send just the value $v.\mathit{val}$; only if any node needs to check the \texttt{AdoptionCriteria} then nodes would exchange their proofs $v.\mathit{proof}$.
With the same Assumptions~1--3 and \texttt{AdoptionCriteria} as in Section~\ref{sec:oblivious}, the above described proof aware protocol is listed in Algorithm~\ref{OBBC:impl}.

%
\begin{algorithm}[t]
	\caption{Proof Aware Variant - code for $p$}
	\label{OBBC:impl}
	\footnotesize
	\begin{algorithmic}[1]
		\Procedure{ComAwareConsensus$_{v,\mathtt{Consensus},\mathtt{AdoptionCriteria}}$}{$v'$}
		\State $\mathit{broadcast}(v'.\mathit{val})$ \label{OBBC:2}
		\State \textbf{wait until} $n-f$ proposals $\mathit{\widehat{val}}$ have been received \label{OBBC:3}
		\State $\mathit{votes}=$ \{received proposals\}
		\If {$\mathit{votes}=\{v.\mathit{val}\}$} 
		\State Decide $v.\mathit{val}$ \label{OBBC:5}
		\State Return $v.\mathit{val}$
		\EndIf 
		\State $\mathit{fullvals}=\{v'\}$ \Comment{couldn't terminate quickly; full protocol}
		\State $\mathit{broadcast}(v')$ \label{OBBC:7} \Comment{exchange also the proofs}
		\State \textbf{wait until} $|\mathit{fullvals}| = n-f$ \label{OBBC:8} \Comment{see also line~\ref{OBBC:14}}
		\If {$\mathit{fullvals}$ matches \texttt{AdoptionCriteria}} \label{OBBC:8}
		\State Return \texttt{Consensus}.Propose($v$).$\mathit{val}$ \label{OBBC:9}
		\Else
		\State Return \texttt{Consensus}.Propose($v'$).$\mathit{val}$ \label{OBBC:11}
		\EndIf 
		\EndProcedure
		\\
		\textbf{upon} receiving $(\hat{v})$ from $q$ \textbf{do}: \label{OBBC:14} 
		\State \indent $\mathit{fullvals} = \mathit{fullvals} \cup \{\hat{v}\}$
		\State {\indent send($v'$)} to $q$ \label{OBBC:15}   
		\\
		\If {invoked Decide $v.\mathit{val}$ and some node invoked \texttt{Consensus}.Propose($\hat{v}$)}
		\State \texttt{Consensus}.Propose($v$) \Comment{this is executed at most once} \label{OBBC:11}
		\EndIf
	\end{algorithmic}
\end{algorithm}

As can be seen, the only difference between Algorithm~\ref{alg:gen} and Algorithm~\ref{OBBC:impl} is that nodes initially only exchange the value part of their proposal.
Only if it is not enough to decide, then they also exchange the proofs and invoke the base consensus protocol accordingly.
The proof is essentially the same as in Section~\ref{sec:oblivious}.

\section{Lower Bound}
\label{sec:lower}

In this section, we show that for classical validity, even when there is a preferred value $v$, there is no asynchronous Byzantine consensus protocol that always terminates in a single round when all nodes start with the preferred value $v$ and there are no Byzantine failures.

\begin{theorem}
	\label{thm:lower}
	When $n=3f+1$, there does not exist an asynchronous Byzantine consensus protocol that always terminate in a single communication round for nodes that receive messages from only $n-f$ nodes (or fewer) in that round and no Byzantine failures manifest even when there is a preferred value.
\end{theorem}

\begin{figure*}[t]
	\center{
		\includegraphics[scale=0.46]{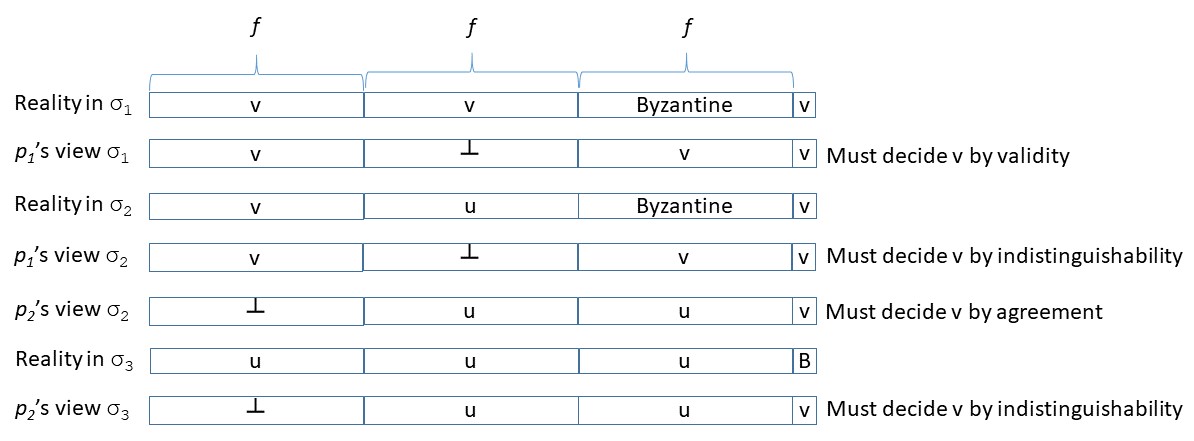}
	}
	\caption{Setup for lower bound proof.}
	\label{fig:lower}
\end{figure*}

\begin{proof}
	Assume by way of contradiction that such a protocol $\mathcal P$ exists.
	Consider an execution $\sigma_1$ of $\mathcal P$ in which $n-f = 2f+1$ correct nodes start with a value $v$ while $f$ additional nodes are Byzantine whom are started with $v$, as illustrated in Figure~\ref{fig:lower}.
	Consider a correct node $p_1$ that receives at most $f+1$ messages from correct nodes and at most $f$ messages from the Byzantine nodes, yet any message potentially sent by the other nodes is not received by $p_1$ until after time $t_1^1$ to be defined shortly.
	The $f$ Byzantine processes send to $p_1$ the same messages they would have generated in $\mathcal P$ had they been started with $v$ (so Byzantine failures do not manifest in this execution).
	Since by assumption when executing $\mathcal P$ node $p_1$ must eventually decide without receiving additional messages, then $p_1$ must ultimately decide at time $t_1^0$ and we set $t_1^0 < t_1^1$.
	As all correct nodes started with $v$ in $\sigma_1$, $p_1$ must have decided $v$ in $\sigma_1$.
	
	Next, consider an execution $\sigma_2$ in which all processes that $p_1$ did not receive their messages until time $t_1^1$ in $\sigma_1$ are started with value $u$, the Byzantine nodes send to $p_1$ exactly the same messages as in $\sigma_1$, and up until time $t_1^1$ all timings are identical to $\sigma_1$.
	Hence, up until time $t_1^0$ node $p_1$ cannot distinguish between $\sigma_1$ and $\sigma_2$ and therefore must decide $v$ in $\sigma_2$ as well.
	
	Further, in execution $\sigma_2$ there is another correct node $p_2$ such that $p_2$ does not receive any of the messages from a subset of $f$ correct nodes that includes $p_1$ until after some time $t_2^1$ to be defined shortly.
	The messages sent to $p_2$ by the Byzantine nodes match exactly their prescribed behavior in $\mathcal P$ had they been started with $u$.
	Since we assume that $\mathcal P$ is an asynchronous Byzantine agreement protocol, then eventually $p_2$ must decide at some time $t_2^0$, and we set $t_2^0 < t_2^1$.
	Since $p_1$ has decided $v$ in $\sigma_2$, by Byzantine agreement $p_2$ must also decide $v$ in~$\sigma_2$.
	
	Finally, consider a third execution $\sigma_3$ of $\mathcal P$ in which one of the $f+1$ correct nodes that started with $v$ in $\sigma_2$ is Byzantine in $\sigma_3$, while all other nodes are correct and start with $u$.
	Further, all timings are the same between $\sigma_2$ and $\sigma_3$ up until time $t_2^1$.
	Hence, $p_2$ cannot distinguish between $\sigma_2$ and $\sigma_3$ up to time $t_2^0$ and therefore must decide $v$ at time $t_2^0$.
	However, all correct processes proposed $u$ in $\sigma_3$, a violation of classical validity.
	A contradiction.
\end{proof}

Notice that the above lower bound does not hold if we allow nodes in Line~\ref{OC:3} of Algorithm~\ref{alg:gen} to wait on a timeout and at least $n-f$ votes.
In such a case, a correct solution with $f<n/3$ is to decide on the value $v$ only if a node collects $3f+1$ votes and all received votes are the same $v$ while the adoption criteria would be to adopt the value $v$ if $f+1$ values are $v$.
Yet, this solution only works when the network is in a synchronous period and there is a known upper bound for the usual message latency.
Also, it requires receiving messages from all nodes, meaning that it always works at the paste of the slowest node.
Further, advancement to the base consensus protocol in case a fast decision cannot be made always waits for the worst case expected network latency.

\section{Discussion}
\label{sec:discuss}

In this work, we explored the communication benefits of biasing consensus into preferring a specific decision value whenever that value can be decided on while preserving the standard consensus correctness requirements.
To that end, we have presented a generic optimizer code that can take an arbitrary consensus protocol and optimize it so that whenever some favorable conditions occur, then the protocol terminates in a single all-to-all communication round.
Locally, the favorable conditions are that the first $n-f$ values a given node receives are all the same preferred value $v$.
Globally, if this occurs for all correct processes, then all correct processes terminate quickly and efficiently.
This is regardless of any timing assumptions and reliability assumptions (other than that the network cannot generate or modify messages), and regardless of any oracles etc.

In practice, our favorable conditions will hold, e.g., whenever all nodes start with the same value and no Byzantine failures manifest.
Yet, up to $f$ benign failures can still occur, and the instantiations of the generic construction to the Byzantine failure mode ensure correctness even when Byzantine failure do occur (when Byzantine failures do occur, this may induce extra communication costs).
As we mentioned before, guessing the preferred consensus value can be done, e.g., in several recent blockchain protocols~\cite{tendermint,BF2019full,Crain2017} (here the biased value is to approve the leader's block proposal), in strongly consistent primary backup replication~\cite{Lamport:1998} (here the biased value is to accept the master's most recent update), and in various control systems~\cite{B777b}.
In case the guessed preferred value cannot be decided on (a ``bad guess''), the only harm is additional communication rounds.

\paragraph*{Acknowledgements} This work was partially funded by ISF grant \#1505/16.

{ \bibliographystyle{abbrv}
	\bibliography{references}

\begin{thebibliography}{10}

\bibitem{AGGRW05}
M.~Abd-El-Malek, G.~R. Ganger, G.~R. Goodson, M.~K. Reiter, and J.~J. Wylie.
\newblock {Fault-scalable Byzantine Fault-tolerant Services}.
\newblock In {\em Proceedings of the 20th ACM Symposium on Operating Systems
  Principles}, SOSP, pages 59--74, 2005.

\bibitem{AW98}
H.~Attiya and J.~Welch.
\newblock {\em {Distributed Computing: Fundamentals, Simulations and Advanced
  Topics (2nd edition)}}.
\newblock John Wiley Interscience, March 2004.

\bibitem{Guerraoui2010}
P.-L. Aublin, R.~Guerraoui, N.~Kne\v{z}evi\'{c}, V.~Qu{\'e}ma, and
  M.~Vukoli\'{c}.
\newblock {The Next 700 BFT Protocols}.
\newblock {\em ACM Trans. Comput. Syst.}, 32(4), Jan. 2015.

\bibitem{FB96}
K.~Birman and R.~Friedman.
\newblock {Trading Consistency for Availability in Distributed Systems}.
\newblock Technical Report TR96-1579, Computer Science Department, Cornell
  University, Apr. 1996.

\bibitem{BBGMR01}
F.~V. Brasileiro, F.~Greve, A.~Most{\'e}faoui, and M.~Raynal.
\newblock {Consensus in One Communication Step}.
\newblock In {\em Proceedings of the 6th International Conference on Parallel
  Computing Technologies}, PaCT, pages 42--50. Springer-Verlag, 2001.

\bibitem{tendermint}
E.~Buchman.
\newblock {Tendermint: Byzantine Fault Tolerance in the Age of Blockchains}.
\newblock Master's thesis, University of Guelph, 2016.

\bibitem{BF2019full}
Y.~Buchnik and R.~Friedman.
\newblock {TOY: a Total ordering Optimistic sYstem for Permissioned
  Blockchains}.
\newblock {\em CoRR}, abs/1901.03279, January 2019.
\newblock Full Version.

\bibitem{cachin-survey}
C.~Cachin and M.~Vukolic.
\newblock {Blockchain Consensus Protocols in the Wild}.
\newblock {\em CoRR}, abs/1707.01873, 2017.

\bibitem{Castro1999}
M.~Castro and B.~Liskov.
\newblock {Practical Byzantine Fault Tolerance}.
\newblock In {\em Proceedings of the 3rd ACM Symposium on Operating Systems
  Design and Implementation}, OSDI, pages 173--186, 1999.

\bibitem{CMLRS}
J.~Cowling, D.~Myers, B.~Liskov, R.~Rodrigues, and L.~Shrira.
\newblock {HQ Replication: A Hybrid Quorum Protocol for Byzantine Fault
  Tolerance}.
\newblock In {\em Proceedings of the 7th Symposium on Operating Systems Design
  and Implementation}, OSDI, pages 177--190. USENIX Association, 2006.

\bibitem{Crain2017}
T.~Crain, V.~Gramoli, M.~Larrea, and M.~Raynal.
\newblock {DBFT: Efficient Leaderless Byzantine Consensus and its Application
  to Blockchains}.
\newblock In {\em 17th {IEEE} International Symposium on Network Computing and
  Applications (NCA)}, 2018.

\bibitem{Fischer1985}
M.~J. Fischer, N.~A. Lynch, and M.~S. Paterson.
\newblock {Impossibility of Distributed Consensus with One Faulty Process}.
\newblock {\em J. ACM}, 32(2):374--382, Apr. 1985.

\bibitem{FriedmanMRR07}
R.~Friedman, A.~Most{\'{e}}faoui, S.~Rajsbaum, and M.~Raynal.
\newblock {Asynchronous Agreement and Its Relation with Error-Correcting
  Codes}.
\newblock {\em {IEEE} Trans. Computers}, 56(7):865--875, 2007.

\bibitem{Friedman:2005}
R.~Friedman, A.~Mostefaoui, and M.~Raynal.
\newblock {Simple and Efficient Oracle-Based Consensus Protocols for
  Asynchronous Byzantine Systems}.
\newblock {\em IEEE Trans. Dependable Secur. Comput.}, 2(1):46--56, Jan. 2005.

\bibitem{KeidarR03}
I.~Keidar and S.~Rajsbaum.
\newblock {On the Cost of Fault-Tolerant Consensus When There Are No Faults - A
  Tutorial}.
\newblock In {\em Dependable Computing, First Latin-American Symposium (LADC)},
  pages 366--368, 2003.

\bibitem{Kotla2007}
R.~Kotla, L.~Alvisi, M.~Dahlin, A.~Clement, and E.~Wong.
\newblock {Zyzzyva: Speculative Byzantine Fault Tolerance}.
\newblock {\em SIGOPS Oper. Syst. Rev.}, 41(6):45--58, Oct. 2007.

\bibitem{Lamport:1998}
L.~Lamport.
\newblock {The Part-time Parliament}.
\newblock {\em ACM Trans. Comput. Syst.}, 16(2):133--169, May 1998.

\bibitem{Lamport1982}
L.~Lamport, R.~Shostak, and M.~Pease.
\newblock {The Byzantine Generals Problem}.
\newblock {\em ACM Trans. Program. Lang. Syst.}, 4(3):382--401, July 1982.

\bibitem{Martin2006}
J.-P. Martin and L.~Alvisi.
\newblock {Fast Byzantine Consensus}.
\newblock {\em IEEE Trans. on Dependable and Secure Computing}, 3(3):202--215,
  July 2006.

\bibitem{Mostefaoui2003}
A.~Mostefaoui, S.~Rajsbaum, and M.~Raynal.
\newblock {Conditions on Input Vectors for Consensus Solvability in
  Asynchronous Distributed Systems}.
\newblock {\em J. ACM}, 50(6):922--954, Nov. 2003.

\bibitem{Raynal2018book}
M.~Raynal.
\newblock {\em {Fault-Tolerant Message-Passing Distributed Systems -- an
  Algorithmic Approach}}.
\newblock Springer International Publishing, 2018.

\bibitem{crypto-textbook}
B.~Schneier.
\newblock {\em {Applied Cryptography : Protocols, Algorithms and Source Code in
  C}}.
\newblock John Wiley \& Sons Inc, December 2015.

\bibitem{B777b}
Y.~Yeh.
\newblock {Safety Critical Avionics for the 777 Primary Flight Controls
  System}.
\newblock In {\em 20th Digital Avionics Systems Conference (DASC)}, 2001.

\end{thebibliography}
}
\end{document}